\documentclass[a4paper,12pt]{article}

\usepackage{amscd}
\usepackage{amsfonts}
\usepackage{amsmath}
\usepackage{amssymb}
\usepackage{amsthm}
\usepackage{color} 
\usepackage[dvipdfm]{graphicx} 
\usepackage[T1]{fontenc} 
\usepackage{here} 
\usepackage{mathrsfs} 
\usepackage{txfonts} 
\usepackage[all]{xy} 
\usepackage{algorithm} 
\usepackage{algpseudocode} 
\usepackage{diagbox} 
\usepackage{ulem} 

\allowdisplaybreaks

\theoremstyle{plain}
\newtheorem{thm}{Theorem}[section]
\newtheorem{lmm}[thm]{Lemma}
\newtheorem{prp}[thm]{Proposition}

\theoremstyle{definition}
\newtheorem{dfn}[thm]{Definition}

\newtheorem{exm}[thm]{Example}

\newcommand{\vs}[1][0.2]{\vspace{#1in}\noindent\ignorespaces}
\newcommand{\ba}{\begin{array*}}
\newcommand{\ea}{\end{array*}}
\newcommand{\be}{\begin{eqnarray*}}
\newcommand{\ee}{\end{eqnarray*}}
\newcommand{\bi}{\begin{itemize}}
\newcommand{\ei}{\end{itemize}}
\newcommand{\bb}{\vs\begin{itembox}}
\newcommand{\eb}{\end{itembox}}
\newcommand{\bc}{\begin{center}}
\newcommand{\ec}{\end{center}}
\newcommand{\bs}{\vs\begin{screen}}
\newcommand{\es}{\end{screen}}

\def\ens#1{{\mathchoice{\left\{ #1 \right\}}{\{ #1 \}}{\{ #1 \}}{\{ #1 \}}}}
\def\set#1#2{{\mathchoice{\left\{ #1 \middle| #2 \right\}}{\{ #1 \mid #2 \}}{\{ #1 \mid #2 \}}{\{ #1 \mid #2 \}}}}
\def\r#1{\text{\rm #1}}

\def\Bigv#1{\left| #1 \right|}
\def\v#1{{\mathchoice{\Bigv{#1}}{| #1 |}{| #1 |}{| #1 |}}}
\def\Bign#1{\left\| #1 \right\|}
\def\n#1{{\mathchoice{\Bign{#1}}{\| #1 \|}{\| #1 \|}{\| #1 \|}}}

\def\tl#1{\tilde{#1}{}}

\def\optimise#1#2{\be \begin{array}{ll} \r{minimise} & #1 \\ \r{subject to} & #2 . \end{array} \ee}

\newcommand{\bC}{\mathbb{C}}

\newcommand{\bN}{\mathbb{N}}

\newcommand{\bQ}{\mathbb{Q}}
\newcommand{\bR}{\mathbb{R}}

\newcommand{\bZ}{\mathbb{Z}}

\newcommand{\cP}{\mathscr{P}}

\newcommand{\C}{\bC}

\newcommand{\N}{\bN}
\newcommand{\Q}{\bQ}
\newcommand{\R}{\bR}
\newcommand{\Z}{\bZ}

\newcommand{\Fp}{\mathbb{F}_p}

\newcommand{\Qp}{\mathbb{Q}_p}

\newcommand{\Zp}{\mathbb{Z}_p}

\newcommand{\CO}{\r{CO}}

\newcommand{\dom}{\r{dom}}

\newcommand{\im}{\r{im}}

\newcommand{\arith}{\r{ar}}
\newcommand{\DNN}{\r{DNN}}
\newcommand{\Graph}{\r{Graph}}
\newcommand{\LR}{\r{LR}}
\newcommand{\MaxCut}{\r{MaxCut}}
\newcommand{\MC}{\r{MC}}
\newcommand{\OP}{\r{OP}}
\newcommand{\SNN}{\r{SNN}}

\algnewcommand\algorithmicbreak{{\bf break}}
\algnewcommand\Break{\algorithmicbreak{}}
\algnewcommand\algorithmiccontinue{{\bf continue}}
\algnewcommand\Continue{\algorithmiccontinue{}}

\title{Generalisation of Baker's Forcing Method to Arbitrary Prime and NP-hardness of Several $p$-adic Optimisations}
\author{Tomoki Mihara}
\date{}

\begin{document}

\maketitle
\begin{abstract}
G.\ D.\ Baker formulated a forcing method to interpret integer optimisation problem into $2$-adic linear regression, and proved the NP-hardness of $2$-adic linear regression. We generalise the forcing method to a wider class of $p$-adic optimisation for the case where $p$ is not necessarily $2$, and prove the NP-hardness of $p$-adic linear regression, the NP-hardness of $2$-adic dynamic neural network by S.\ Albeverio, A.\ Khrennikov, and B.\ Tirrozi, and the NP-hardness of a partial generalisation of the $p$-adic optimisation problem associated to van der Put neural network by G.\ L.\ R.\ N'guessan.
\end{abstract}

\tableofcontents

\section{Introduction}
\label{Introduction}

Let $p$ be a prime number. The notion of $p$-adic numbers is introduced by K.\ Hensel in 1897 in \cite{Hen97}, and plays a central role in modern number theory. Recently, the $p$-adic numbers are used also in other branches of science, because of many significant similarities to and differences from the real numbers. See Introductions of \cite{Bra25} and \cite{Mih26-1} for history and basic knowledge respectively on applications of $p$-adic numbers in computer science.

\vs
S.\ Albeverio, A.\ Khrennikov, and B.\ Tirrozi formulated and studied $2$-adic neural network called {\it $2$-adic dynamic neural network} in \cite{AKT99} and \cite{KT00}. Starting from these innovative studies, people proposed and studied various optimisation problems related to $p$-adic numbers: {\it $p$-adic cellular neural network} by W.\ A.\ Z\'u\~niga-Galindo, B.\ A.\ Zambrano-Luna, and B.\ Dibba in \cite{ZZ23} and \cite{ZZB24}, an alternative of $p$-adic neural network by an explicit choice of a homeomorphism in Brouwer's characterisation theorem (cf.\ \cite{Bro10}) by A.\ P.\ Zubarev in \cite{Zub25-1} and \cite{Zub25-2}, {\it van der Put neural network} by G.\ L.\ R.\ N'guessan in \cite{Ngu25}, and {\it $p$-adic character neural network} by us in \cite{Mih26-3}.

\vs
Since studies of $p$-adic optimisation just have a short history compared to those of real optimisation, not so many were known for $p$-adic optimisation before 2020. For example, E.\ Amaldi and V.\ Kann showed that the maximal feasible subsystem problem of linear equations over the finite field $\Fp$ is APX-complete, i.e.\ complete for the class of optimisation problems which allow constant-factor approximations, in \cite{AK95} Proposition A.1. We note that the optimisation problem of modulo $p$ linear regression for the $\ell^1$-norm of the trivial valuation on $\Fp$ is identical to the maximal feasible subsystem problem of linear equations over $\Fp$ and is included in the optimisation problem of $p$-adic linear regression for $\ell^q$-norm of the $p$-adic valuation on $\Qp$ for $q \in (0,\infty)$. See also Introduction of \cite{Mih26-4} for other preceding studies on $p$-adic optimisation.

\vs
However, as if they synchronised with the evolution of various studies on $p$-adic neural networks, many results on $p$-adic optimisation appeared after 2020. For example, A.\ F.\ T.\ Martins showed polynomial-time algorithms for $p$-adic linear regression for $\ell^{\infty}$-norm in \cite{Mar25}. G.\ D.\ Baker, S.\ Mccallum, and D.\ Pattinson gave a lower bound of the number of sample points in a hyperplane optimal of $p$-adic linear regression for $\ell^1$-norm in \cite{BMP25}. Further, G.\ D.\ Baker proved that the class of $p$-adic linear regressions for $\ell^1$-norm with varying $p$ is NP-hard in the sense that the maximum cut problem can be reduced to a specific instance of $2$-adic linear regressions in \cite{Bak26}. We formulated a new $p$-adic optimisation problem as an infinitesimal limit of the least squares method, invented and compared various heuristic algorithms for $p$-adic polynomial regression with thorough analysis of their time and space complexity in \cite{Mih26-1}. Further, we formulated {\it $p$-adic principal component analyses} as dimensionality reduction methods in $p$-adic vector spaces in \cite{Mih26-2}. Moreover, we introduced a new probabilistic algorithm of $p$-adic linear regression in \cite{Mih26-4}. Furthermore, we introduced a new heuristic algorithm for $p$-adic manifold learning and propose a $p$-adic benchmark task in \cite{Mih26-5}.

\vs
In this paper, we formulate the notion of NP-hardness specialised to a class of $p$-adic optimisation problems, and prove the NP-hardness of $p$-adic linear regression as an extension of the result in \cite{Bak26}. The fundamental tool is an abstract generalisation of forcing method in \cite{Bak26}, and hence the proof works in a parallel way. Similarly, we prove the NP-hardness of $p$-adic dynamic neural network, and the NP-hardness of a partial generalisation of the $p$-adic optimisation problem associated to van der Put neural network, which we call {\it $p$-adic smooth neural network}.

\vs
We briefly explain contents of this paper. In \S \ref{Convention}, we introduce convention for this paper. In \S \ref{p-adic Optimisation}, we introduce $p$-adic linear regression and $2$-adic dynamic neural network, and formulate $p$-adic smooth neural network. In \S \ref{Forcing Method}, we formulate the generalisation of Baker's forcing method. In \S \ref{NP-hardness}, we formulate the notion of NP-hardness of a class of $p$-adic optimisation problems, and show the NP-hardness of $p$-adic linear regression, $p$-adic dynamic neural network, and $p$-adic smooth neural network.

\section{Convention}
\label{Convention}

We denote by $\N$ the set of non-negative integers. For a set $X$, we denote by $\cP(X)$ the set of subsets of $X$. For a set $X$, we denote by $\# X$ the cardinality of $X$. For a set $X$, an $x_0 \in X$, and a binary relation $R$ on $X$, we set $X_{R x_0} \coloneqq \set{x \in X}{x R x_0}$. For sets $X$ and $Y$, we denote by $X^Y$ the set of maps $Y \to X$. We note that every $n \in \N$ is identified with $\N_{< n}$ in set theory, and hence $X^n$ formally means $X^{\N_{< n}}$, which is naturally identified with the set of $n$-tuples in $X$. For a map $f$, we denote by $\dom(f)$ the domain of $f$.

\vs
Throughout the paper, $p$ denotes a fixed prime number. We denote by $\Qp$ the field of $p$-adic numbers, and by $\Zp$ the ring of $p$-adic integers. We fix a $p$-adic absolute value $\v{\cdot} \colon \Qp \to \R_{\geq 0}$. We note that the $p$-adic absolute value with the normalisation $\v{p} = p^{-1}$ is frequently denoted by $\v{\cdot}_p$, but we do not use the convention because we do not deal with the Euclidean absolute value $\C \to \R_{\geq 0}$.

\vs
For an $(n,m) \in \Z \times \N_{> 0}$, we denote by $n \bmod m$ the unique element of $\N_{< m}$ congruent to $n$ modulo $m$. Let $q \in \N$ be a power of $p$. For an $x \in \Zp$, we denote by $x \bmod q \in \N_{< q}$ the unique element of $\N_{< q}$ congruent to $x$ modulo $q$. Although $n \bmod q$ is multiply defined for an $n \in \Z$, there is no ambiguity because the two results coincide with each other. For a $\vec{x} = (x_i)_{i=0}^{n-1} \in \Zp^n$ with $n \in \N$, we set $\vec{x} \bmod q \coloneqq (x_i \bmod q)_{i=0}^{n-1}$. For an $X = (\vec{x}_i)_{i=0}^{n-1} \in (\Zp^m)^n$ with $(n,m) \in \N^2$, we set $X \bmod q \coloneqq (\vec{x}_i \bmod q)_{i=0}^{n-1}$.

\vs
For a set $X$ and a $U \in \cP(X)$, we abuse the convention $1_U$ to mean the characteristic function $X \to \Qp$ of $U$ without specifying $X$ and $p$. For example, $1_{\ens{0,1}}$ regarded as an element of $\Qp^3 = \Qp^{\N_{< 3}}$ stands for the vector $(1,1,0)$.

\section{$p$-adic Optimisation}
\label{p-adic Optimisation}

A {\it $p$-adic optimisation} is a map $\ell \colon \Qp^n \to \R_{\geq 0} \sqcup \ens{\infty}$ with $n \in \N$. Let $\ell$ be a $p$-adic optimisation. Set $\dim(\ell) \coloneqq \dim_{\Qp} \dom(\ell)$. For a $C \in \cP(\dom(\ell))$, we denote by $\OP(\ell,C)$ the following optimisation problem:
\optimise{\ell(\vec{c})}{\vec{c} \in C}
We abbreviate $\OP(\ell,\dom(\ell))$ to $\OP(\ell)$. We give three examples of $p$-adic optimisations.

\begin{exm}[$p$-adic linear regression]
For an $X = (\vec{x}_i)_{i=0}^{N-1} \in (\Qp^D)^N$ and a $\vec{y} = (y_i)_{i=0}^{N-1} \in \Qp^N$ with $(N,D) \in \N^2$, we denote by $\ell^{\LR}_{X,\vec{y}}$ the $p$-adic optimisation
\be
\Qp^D & \to & \R_{\geq 0} \sqcup \ens{\infty} \\
\vec{c} & \mapsto & \sum_{i=0}^{N-1} \v{y_i - \vec{c} \cdot \vec{x}_i},
\ee
where $\cdot$ denotes the standard inner product on $\Qp^D$, and call it {\it the $p$-adic linear regression} with input data $X$ and output data $\vec{y}$.
\end{exm}

\begin{exm}[$2$-adic dynamical neural network]
Let $E \in \N$. Let $(n,m) \in \N^2$, $A = ((A_{i,j})_{j=0}^{m-1})_{i=0}^{n-1} \in (\Q_2^m)^n$, and $\vec{b} = (b_j)_{j=0}^{m-1} \in \Q_2^m$. We denote by $f^{\DNN}_{E,A,\vec{b}}$ the map
\be
\Q_2^n & \to & \Q_2^m \\
(x_i)_{i=0}^{n-1} & \mapsto & \left( 1 - 1_{2^E \Z_2} \left( \left( \sum_{i=0}^{n-1} x_i A_{i,j} \right) - b_j \right) \right)_{j=0}^{m-1}.
\ee
For an $X = (\vec{x}_{\alpha})_{\alpha=0}^{N-1} \in (\ens{0,1}^n)^N$ and a $Y = (\vec{y}_{\alpha})_{\alpha=0}^{N-1} \in (\ens{0,1}^m)^N$ with $N \in \N$, we denote by $\ell^{\DNN}_{E,X,Y}$ the $2$-adic optimisation
\be
\Q_2^{m(n+1)} \cong (\Q_2^m)^n \times \Q_2^m & \to & \R_{\geq 0} \sqcup \ens{\infty} \\
(A,\vec{b}) & \mapsto & 
\left\{
\begin{array}{ll}
\sum_{\alpha=0}^{N-1} \n{\vec{y}_{\alpha} - f^{\DNN}_{E,A,\vec{b}}(\vec{x}_{\alpha})} & ((A,\vec{b}) \in (\Z_2^m)^n \times \Z_2^m) \\
\infty & ((A,\vec{b}) \notin (\Z_2^m)^n \times \Z_2^m)
\end{array}
\right.,
\ee
where $\n{\cdot}$ denotes the $\ell^1$-norm on $\Q_2^m$ with respect to the $2$-adic valuation $\v{\cdot}$, and call it {\it the $2$-adic dynamical neural network} with input data $X$, output data $Y$, and precision $E$. We note that the original definition in \cite{KT00} \S 2 rescales $\ell^{\DNN}_{E,X,Y}$ by $N^{-1}$. Also we note that the study of learning process in \cite{KT00} \S 2 implicitly focuses only on the case $\vec{b} = (0)_{j=0}^{m-1}$ from the middle of the discussion, but we consider the full generality.
\end{exm}

\begin{exm}[van der Put neural network]
Let $D \in \N$. We denote by $\CO(\Zp^D)$ the set of clopen subsets of $\Zp^D$, and by $B_{p,D}$ the set of maps $B \colon \N_{< M} \to \CO(\Zp^D)$ with $M \in \N$. Let $B \in B_{p,D}$. Set $M \coloneqq \# \dom(B) \in \N$. For an $X = (\vec{x}_i)_{i=0}^{N-1} \in (\Qp^D)^N$ and a $\vec{y} = (y_i)_{i=0}^{N-1} \in \Qp^N$, we denote by $\ell^{\SNN}_{B,X,\vec{y}}$ the $p$-adic optimisation
\be
\Qp^M & \to & \R_{\geq 0} \sqcup \ens{\infty} \\
\vec{c} = (c_j)_{j=0}^{M-1} & \mapsto & \sum_{i=0}^{N-1} \v{y_i - \sum_{j=0}^{M-1} c_j 1_{B(j)}(\vec{x}_i)},
\ee
and call it {\it the $p$-adic smooth neural network} with input data $X$, output data $\vec{y}$, and support $B$. When $D = 1$ and $B$ is an initial segment of the van der Put basis, then it is identical to the optimisation problem associated to the van der Put neural network with $\alpha = 0$ in \cite{Ngu25} (although $\alpha = 0.01$ throughout the paper other than Theorem A.1). Therefore, $p$-adic smooth neural network is a partial generalisation of the optimisation problem associated to van der Put neural network.
\end{exm}

\section{Forcing Method}
\label{Forcing Method}

Let $\ell$ be a $p$-adic optimisation. For an $E \in \N$ and a $C \in \cP(\N_{< p^E}^{\dim(\ell)})$, we say that $\ell$ is {\it $(E,C)$-reducible} if it satisfies the following:
\bi
\item[(1)] An optimal solution of $\OP(\ell,C)$ is an optimal solution of $\OP(\ell)$.
\item[(2)] For any optimal solution $\vec{c} \in \dom(\ell)$ of $\OP(\ell)$, $\vec{c}$ belongs to $C + p \Zp^{\dim(\ell)}$ and $\vec{c} \bmod p^E$ is an optimal solution of $\OP(\ell,C)$.
\ei
By the definition, $(E,C)$-reducibility implies $(E,C')$-reducibility for any $E \in \N$ and any $(C,C') \in \cP(\N_{< p^E}^{\dim(\ell)})^2$ with $C \subset C'$.

\begin{prp}
\label{optimality restriction}
For any $E \in \N$ and any $C \in \cP(\N_{< p^E}^{\dim(\ell)})$, if the following hold, then $\ell$ is $(E,C)$-reducible:
\bi
\item[(1)] For any $\vec{c} \in \dom(\ell) \setminus (C + p^E \Zp^{\dim(\ell)})$, there exists a $\vec{d} \in C$ such that $\ell(\vec{d}) < \ell(\vec{c})$.
\item[(2)] Every $\vec{c} \in C + p^E \Zp^{\dim(\ell)}$ satisfies $\ell(\vec{c} \bmod p^E) \leq \ell(\vec{c})$.
\ei
\end{prp}

\begin{proof}
For any $\vec{c} \in \dom(\ell)$, there exists a $\vec{d} \in C$ such that $\ell(\vec{d}) \leq \ell(\vec{c})$ by the conditions (1) and (2) in the assertion. This implies the condition (1) in the definition of $(E,C)$-reducibility.

\vs
Let $\vec{c} \in \dom(\ell)$ be an optimal solution of $\OP(\ell)$. By the condition (1) in the assertion, we have $\vec{c} \in C + p^E \Zp^{\dim(\ell)}$, i.e.\ $\vec{c} \bmod p^E$ belongs to $C$. By the optimality of $\vec{c}$ for $\OP(\ell)$ and the condition (2), we have $\ell(\vec{c}) = \ell(\vec{c} \bmod p^E)$, and hence $\ell(\vec{c} \bmod p^E)$ and is an optimal solution of $\OP(\ell,C)$.
\end{proof}

Let $\vec{\ell} = (\ell_i)_{i=0}^{N-1}$ be a sequence with an $N \in \N$ of $p$-adic optimisations such that $\dim(\ell_i)$ is independent of $i \in \N_{< N}$. We define $\dim(\vec{\ell})$ as $\dim(\ell_0)$ if $N > 0$ and $0$ otherwise, and $\Sigma(\vec{\ell})$ as the $p$-adic optimisation
\be
\Qp^{\dim(\vec{\ell})} & \to & \R_{\geq 0} \sqcup \ens{\infty} \\
\vec{c} & \mapsto & \sum_{i=0}^{N-1} \ell_i(\vec{c}).
\ee
For a $C \in \cP(\N_{< p}^{\dim(\vec{\ell})})$, we say that $\vec{\ell}$ is {\it concentrated on $C$} if for any $(i,\vec{c}) \in \N_{< N} \times C + p \Zp^{\dim(\vec{\ell})}$, the following are equivalent:
\bi
\item[(1)] $\ell_i(\vec{c}) \geq 1$
\item[(2)] $\ell_i(\vec{c} \bmod p) = 1$
\item[(3)] $\ell_i(\vec{c} \bmod p) > 0$
\ei
The following is a generalisation of Baker's forcing method for $2$-adic linear regression appearing in the proof of \cite{Bak26} Theorem 3.1:

\begin{thm}
\label{Baker}
Let $\vec{B} = (B_d)_{d=0}^{\dim(\vec{\ell})-1} \in (\cP(\N_{< p}) \setminus \ens{\emptyset})^{\dim(\vec{\ell})}$ and $M \in \N$. Set $C \coloneqq \prod_{d=0}^{\dim(\vec{\ell})-1} B_d$. Suppose that $\sup \Sigma(\vec{\ell})(C) < M$ and $\vec{\ell}$ is concentrated on $C$. Let $w \colon \Qp \to \R_{\geq 0}$ denote either $\v{\cdot}$ or $\v{1_{\Qp \setminus p \Zp}(\cdot)}$. We denote by $\ell$ the $p$-adic optimisation
\be
\Qp^{\dim(\vec{\ell})} & \to & \R_{\geq 0} \sqcup \ens{\infty} \\
\vec{c} = (c_d)_{d=0}^{\dim(\vec{\ell})-1} & \mapsto & \Sigma(\vec{\ell}) + M \sum_{d=0}^{\dim(\vec{\ell})-1} \sum_{b \in B_d} w(c_d - b).
\ee
Then $\ell$ is $(1,C)$-reducible, and a $\vec{c} \in C$ is an optimal solution of $\OP(\ell)$ if and only if it is an optimal solution of $\OP(\Sigma(\vec{\ell}),C)$.
\end{thm}

\begin{proof}
The proof is essentially the same as that of \cite{Bak26} Theorem 3.1. Since $\vec{B}$ is a sequence of non-empty sets, $C$ is non-empty. Taking a $\vec{c}_0 \in C$, we have
\be
\ell(\vec{c}_0) & = & \Sigma(\vec{\ell})(\vec{c}_0) + M \sum_{d=0}^{\dim(\vec{\ell})-1} (\# B_d - 1) \leq \sup \Sigma(\vec{\ell})(C) + M \sum_{d=0}^{\dim(\vec{\ell})-1} (\# B_d - 1) \\
& < & M \left( 1 + \sum_{d=0}^{\dim(\vec{\ell})-1} (\# B_d - 1) \right).
\ee
Let $\vec{c} \in \Qp^{\dim(\vec{\ell})}$. If $\vec{c} \notin C + p \Zp^{\dim(\vec{\ell})}$, then we have
\be
\ell(\vec{c}) = \Sigma(\vec{\ell})(\vec{c}) + M \left( 1 + \sum_{d=0}^{\dim(\vec{\ell})-1} (\# B_d - 1) \right) > \ell(\vec{c}_0).
\ee
Suppose $\vec{c} \in C + p \Zp^{\dim(\vec{\ell})}$. Set $\vec{b} = (b_d)_{d=0}^{\dim(\vec{\ell})-1} \coloneqq \vec{c} \bmod p \in C$. Since $\vec{\ell}$ is concentrated on $C$, we have $\Sigma(\vec{\ell})(\vec{b}) \leq \Sigma(\vec{\ell})(\vec{c})$. For any $d \in \N_{< \dim(\vec{\ell})}$ and any $b \in B_d$, we have $w(b_d - b) \leq w(c_d - b)$. Therefore, we obtain
\be
\ell(\vec{c}) & = & \Sigma(\vec{\ell})(\vec{c}) + M \sum_{d=0}^{\dim(\vec{\ell})-1} \sum_{b \in B_d} w(c_d - b) \geq \Sigma(\vec{\ell})(\vec{b}) + M \sum_{d=0}^{\dim(\vec{\ell})-1} \sum_{b \in B_d} w(b_d - b) \\
& = & \ell(\vec{b}).
\ee
This implies that $\ell$ is $(1,C)$-reducible by Proposition \ref{optimality restriction}.

\vs
Let $\vec{c} \in C$. We have
\be
\sum_{d=0}^{\dim(\vec{\ell})-1} \sum_{b \in B_d} w(c_d - b) = \sum_{d=0}^{\dim(\vec{\ell})-1} (\# B_d - 1).
\ee
Therefore the optimality of $\vec{c}$ for $\OP(\ell,C)$ is equivalent to that for $\OP(\Sigma(\vec{\ell}),C)$. Since $\ell$ is $(1,C)$-reducible, the optimality of $\vec{c}$ for $\OP(\ell,C)$ is equivalent to that for $\OP(\ell)$. Therefore, the optimality of $\vec{c}$ for $\OP(\ell)$ is equivalent to that for $\OP(\Sigma(\vec{\ell}),C)$.
\end{proof}

We call $\ell$ in Theorem \ref{Baker} a {\it Baker forcing of $\vec{\ell}$ to $\vec{B}$}.

\section{NP-hardness}
\label{NP-hardness}

In order to discuss computational complexity of a $p$-adic optimisation, we formulate reduction of $p$-adic problems to arithmetic problems.

\vs
Let $S$ be a set of $p$-adic optimisations. We denote by $S^{\arith}$ the set of $(\ell,E) \in S \times \N$ such that $\ell$ is $(E,\N_{< p^E}^{\dim(\ell)})$-reducible, and set $\OP^{\arith}(S) \coloneqq \set{\OP(\ell,\N_{< p^E}^{\dim(\ell)})}{(\ell,E) \in S^{\arith}}$. Since $\OP^{\arith}(S)$ is a set of integer optimisation problems, which are canonically translated into integer decision problems, its NP-hardness makes sense.

\begin{dfn}
We say that $S$ is {\it NP-hard} if $\OP^{\arith}(S)$ is NP-hard.
\end{dfn}

We defined the notion of NP-hardness for a set of $p$-adic optimisations for a fixed prime number $p$, while \cite{Bak26} deals the NP-hardness of $p$-adic linear regression for varying prime numbers $p$. However, the difference of the formulation does not matter when we discuss the NP-hardness of a set of $p$-adic optimisations for a fixed prime number $p$, because our formulation just assumes a stronger condition using reducibility. In particular, the following is an extension of \cite{Bak26} Theorem 3.1:

\begin{thm}[NP-hardness of $p$-adic linear regression for each individual $p$]
\label{NP-hardness of LR}
For any prime number $p$, the $p$-adic linear regression is NP-hard, i.e.\ the set
\be
\set{\ell^{\LR}_{X,\vec{y}}}{(D,N) \in \N^2 \land (X,\vec{y}) \in (\Qp^D)^N \times \Qp^N}
\ee
of $p$-adic optimisations is NP-hard.
\end{thm}

In order to prove Theorem \ref{NP-hardness of LR}, we prepare convention and a lemma. We denote by $\Graph$ the set of pairs $(n,e)$ of an $n \in \N$ and an injective map $e \colon \N_{< m} \hookrightarrow \set{(i_0,i_1) \in \N_{< n}^2}{i_0 < i_1}$ with $m \in \N$. For a $G = (n,e) \in \Graph$, we denote by $\MaxCut(G)$ the maximum cut problem for the unweighted undirected graph associated to $G$:
\optimise{\# \set{j \in \N_{< m}}{e(j) \in S^2 \sqcup (\N_{< n} \setminus S)^2}}{S \in \cP(\N_{< n})}
By \cite{GJS76} Theorem 1.2, the set $\MaxCut(\Graph)$ of integer optimisation problems is NP-hard. We show the NP-hardness of various classes of $p$-adic optimisations by reduction to $\MaxCut(\Graph)$.

\vs
For a $G = (n,e) \in \Graph$, we denote by $\ell^{\MC}_{p,G}$ the $p$-adic optimisation
\be
\Qp^n & \to & \R_{\geq 0} \sqcup \ens{\infty} \\
\vec{c} = (c_i)_{i=0}^{n-1} & \mapsto & \sum_{j=0}^{m-1} \v{1 - \sum_{i \in \im(E(j))} c_i} + (m+1) \sum_{i=0}^{n-1} \sum_{b=0}^{1} \v{b - c_i},
\ee
which appeared in the proof of \cite{Bak26} Theorem 3.1 under the additional assumption $p = 2$.

\begin{lmm}
\label{MaxCut to LR}
For any $G = (n,e) \in \Graph$, $\ell^{\MC}_{p,G}$ is $(1,\ens{0,1}^n)$-reducible, and a $\vec{c} \in \ens{0,1}^n$ is an optimal solution of $\OP(\ell^{\MC}_{p,G})$ if and only if $\set{i \in \N_{< n}}{c_i = 1}$ is an optimal solution of $\MaxCut(G)$.
\end{lmm}

\begin{proof}
For each $j \in \N_{< m}$, we define a $p$-adic optimisation $\ell_j$ as
\be
\Qp^n & \to & \R_{\geq 0} \sqcup \ens{\infty} \\
(c_i)_{i=0}^{n-1} & \mapsto & \v{1 - \sum_{i \in \im(E(j))} c_i}.
\ee
Set $\vec{\ell} \coloneqq (\ell_j)_{j=0}^{m-1}$ and $C \coloneqq \ens{0,1}^n$. We have $\sup \Sigma(\vec{\ell})(C) \leq m < m + 1$.

\vs
Let $(j,\vec{c}) \in \N_{< m} \times C + p \Zp^{\dim(\vec{\ell})}$. We show that the following are equivalent:
\bi
\item[(1)] $\ell_j(\vec{c}) \geq 1$
\item[(2)] $\ell_j(\vec{c} \bmod p) = 1$
\item[(3)] $\ell_j(\vec{c} \bmod p) > 0$
\ei
Set $(i_0,i_1) \coloneqq e(j)$, $(c_i)_{i=0}^{n-1} \coloneqq \vec{c}$, and $(b_i)_{i=0}^{n-1} \coloneqq \vec{c} \bmod p$. We have
\be
\ell_j(\vec{c}) & = & \v{1 - c_{i_0} - c_{i_1}} \leq 1 \\
\ell_j(\vec{c} \bmod p) & = & \v{1 - b_{i_0} - b_{i_1}} \in \ens{0,1}.
\ee
The condition (1) implies the condition (2) by $\v{(c_{i_0} + c_{1_1}) - (b_{i_0} + b_{i_1})} < 1$. The condition (2) implies the condition (3) by $1 > 0$. The condition (3) implies the condition (1) by $\ell_j(\vec{c} \bmod p) \in \ens{0,1}$ and $\v{(c_{i_0} + c_{1_1}) - (b_{i_0} + b_{i_1})} < 1$. Therefore, $\vec{\ell}$ is concentrated on $C$.

\vs
For any $\vec{c} = (c_i)_{i=0}^{n-1} \in C$, setting $S \coloneqq \set{i \in \N_{< n}}{c_i = 1}$, we have
\be
\Sigma(\vec{\ell})(\vec{c}) = \sum_{j=0}^{m-1} \v{1 - \sum_{i \in \im(E(j))} c_i} = \# \set{j \in \N_{< m}}{e(j) \in S^2 \sqcup (\N_{< n} \setminus S)^2},
\ee
and hence $\vec{c}$ is an optimal solution of $\OP(\Sigma(\vec{\ell}),C)$ if and only if $S$ is an optimal solution of $\MaxCut(G)$. Since $\ell^{\MC}_{p,G}$ is a Baker forcing of $\vec{\ell}$ to $C$, the assertion follows from Theorem \ref{Baker}.
\end{proof}

\begin{proof}[Proof of Theorem \ref{NP-hardness of LR}]
Set
\be
S & \coloneqq & \set{\ell^{\MC}_{p,G}}{G \in \Graph} \\
S' & \coloneqq & \set{\ell^{\LR}_{X,\vec{y}}}{(D,N) \in \N^2 \land (X,\vec{y}) \in (\Qp^D)^N \times \Qp^N}.
\ee
By Lemma \ref{MaxCut to LR}, we have $\set{(\ell^{\MC}_{p,G},1)}{G \in \Graph} \subset S^{\arith}$ and $\MaxCut(\Graph)$ is reduced to $\OP^{\arith}(S)$. Therefore, the NP-hardness of $\MaxCut(\Graph)$ implies that of $S$.

\vs
Let $G = (n,e) \in \Graph$. We denote by $m \in \N$ the cardinality of the domain of $e$. Set $M \coloneqq m + 2(m+1)n$. For each $j \in \N_{< M}$, we denote by $\vec{x}_{p,G,j} \in \Qp^n = \Qp^{\N_{< n}}$ the characteristic function of $\im(e(j)) \subset \N_{< n}$ if $j < m$ and of $\ens{(j-m) \bmod n} \subset \N_{< n}$ if $j \geq m$, and define $y_{p,G,j} \in \Qp$ as $1$ if $j < m$ and $\lfloor n^{-1}(j-m) \rfloor \bmod 2$ if $j \geq m$. Set $X_{p,G} \coloneqq (\vec{x}_{p,G,j})_{j=0}^{M-1} \in (\Qp^n)^{M}$ and $\vec{y}_{p,G} \coloneqq (y_{p,G,j})_{j=0}^{M-1} \in \Qp^{M}$. We have $\ell^{\MC}_{p,G} = \ell^{\LR}_{X_{p,G},\vec{y}_{p,G}} \in S'$.

\vs
We obtain $S \subset S'$. Therefore, the NP-hardness of $S$ implies that of $S'$.
\end{proof}

\begin{thm}
\label{NP-hardness of DNN}
The $2$-adic dynamic neural network with precision $1$ is NP-hard, i.e.\ the set
\be
\set{\ell^{\DNN}_{1,X,Y}}{(n,m,N) \in \N^3 \land (X,Y) \in (\Q_2^n)^N \times (\Q_2^m)^N}
\ee
of $2$-adic optimisations is NP-hard.
\end{thm}

In order to prove Theorem \ref{NP-hardness of DNN}, we prepare convention and a lemma. Let $G = (n,e) \in \Graph$. We denote by $m \in \N$ the cardinality of the domain of $e$. For each $j \in \N_{< 2m+1}$, we denote by $\vec{x}_{G,j} \in \Q_2^n = \Q_2^{\N_{< n}}$ the characteristic function of $\im(e(j)) \subset \N_{< n}$ if $j < m$ and of $\emptyset \subset \N_{< n}$ if $j \geq m$, and define $\vec{y}_{G,j} \in \Q_2^1 = \Q_2^{\N_{< 1}}$ as $(0)$ if $j < m$ and $(1)$ if $j \geq m$. Set $X_{G} \coloneqq (\vec{x}_{2,G,j})_{j=0}^{2m} \in (\Qp^n)^{2m+1}$ and $Y_{G} \coloneqq (\vec{y}_{G,j})_{j=0}^{2m} \in (\Q_2^1)^{2m+1}$.

\begin{lmm}
\label{MaxCut to DNN}
The $2$-adic optimisation $\ell^{\DNN}_{1,X_{G},Y_{G}}$ is $(1,\ens{0,1}^n \times \ens{1})$-reducible, and for any $A = ((A_i))_{i=0}^{n-1} \in (\ens{0,1}^1)^n$, $(A,(1))$ is an optimal solution of $\OP(\ell^{\DNN}_{1,X_{G},Y_{G}})$ if and only if $\set{i \in \N_{< n}}{A_i = 1}$ is an optimal solution of $\MaxCut(G)$.
\end{lmm}

\begin{proof}
Set $C \coloneqq \ens{0,1}^n \times \ens{1}$. For each $j \in \N_{< m}$, we denote by $\ell_j$ the $2$-adic optimisation
\be
\Q_2^{m(n+1)} \cong (\Q_2^m)^n \times \Q_2^m & \to & \R_{\geq 0} \sqcup \ens{\infty} \\
(A,\vec{b}) & \mapsto & 
\left\{
\begin{array}{ll}
\sum_{j=0}^{m-1} \n{f_{A,\vec{b}}(\vec{x}_j)} & ((A,\vec{b}) \in (\Z_2^1)^n \times \Z_2^1) \\
\infty & ((A,\vec{b}) \notin (\Z_2^1)^n \times \Z_2^1)
\end{array}
\right..
\ee
Set $\vec{\ell} = (\ell_j)_{j=0}^{m-1}$. By $C \subset \Z_2^n \times \Z_2^1$, we have $\sup \Sigma(\vec{\ell})(C) \leq m < m + 1$.

\vs
Let $A = ((A_i))_{i=0}^{n-1} \in (\Q_2^1)^n$ and $\vec{b} \in \Q_2^1$. We denote by $b_0 \in \Z_2$ the unique entry of $\vec{b}$. If $(A,\vec{b}) \notin (\Z_2^1)^n \times \Z_2^1$, then we have
\be
\Sigma(\vec{\ell})(A,\vec{b}) = \ell^{\DNN}_{1,X_{G},Y_{G}}(A,\vec{b}) = \infty > \Sigma(\vec{\ell})(((0))_{i=0}^{n-1},(1)).
\ee
Set $S \coloneqq \set{i \in \N_{< n}}{A_i = 1}$. If $(A,\vec{b}) \in (\Z_2^1)^n \times \Z_2^1$, then for any $j \in \N_{< m}$, setting $(i_0,i_1) \coloneqq e(j)$, we have
\be
\ell_j(A,\vec{b}) & = & \n{(0) - f_{A,\vec{b}}(\vec{x}_j)} = \v{-1 + 1_{2^1 \Z_2} \left( \left( \sum_{i \in \im(e(j))} A_i \right) - b_0 \right)} \\
& = & \v{1_{\Q_2 \setminus 2 \Z_2}(A_{i_0} + A_{i_1} - b_0)} =
\left\{
\begin{array}{ll}
0 & (A_{i_0} \bmod 2 = (A_{i_1} + b_0) \bmod 2) \\
1 & (A_{i_0} \bmod 2 \neq (A_{i_1} + b_0) \bmod 2)
\end{array}
\right.
\ee
Therefore, $\vec{\ell}$ is concentrated on $\ens{0,1}^n \times \ens{1}$, and $(\vec{A},\vec{b})$ is an optimal solution of $\OP(\Sigma(\vec{\ell}),C)$ if and only if $(\vec{A},\vec{b}) \in C$ and $S$ is an optimal solution of $\MaxCut(G)$. We have
\be
\ell^{\DNN}_{1,X_{G},Y_{G}} + (m+1)n & = & \Sigma(\vec{\ell}) + \sum_{j=0}^{m-1} \n{(1) - f_{A,\vec{b}}((0)_{i=0}^{n-1})} + (m+1)n \\
& = & \Sigma(\vec{\ell}) + (m+1) \v{1 - 1_{2^1 \Z_2}(-b_0)} + (m+1)n \\
& = & \Sigma(\vec{\ell}) + (m+1) \left( \v{1_{\Q_2 \setminus 2 \Z_2}(b_0)} + n \right) \\
& = & \Sigma(\vec{\ell}) + (m+1) \left( \v{1_{\Q_2 \setminus 2 \Z_2}(b_0)} + \sum_{i=0}^{n-1} \sum_{b=0}^{1} \v{1_{\Q_2 \setminus 2 \Z_2}(A_i - b)} \right).
\ee
Therefore, $\ell^{\DNN}_{1,X_{G},Y_{G}} + (m+1)n$ is a Baker forcing of $\vec{\ell}$ to $\ens{0,1}^{n} \times \ens{1}$. By Theorem \ref{Baker}, $\ell^{\DNN}_{1,X_{G},Y_{G}}$ is $(1,C)$-reducible, and $(A,\vec{b})$ is an optimal solution of $\OP(\ell^{\DNN}_{1,X_{G},Y_{G}})$ if and only if it is an optimal solution of $\OP(\Sigma(\vec{\ell}),C)$, i.e.\ $(\vec{A},\vec{b}) \in C$ and $S$ is an optimal solution of $\MaxCut(G)$.
\end{proof}

\begin{proof}[Proof of Theorem \ref{NP-hardness of DNN}]
Set $S \coloneqq \set{\ell^{\DNN}_{1,X,Y}}{(n,N) \in \N^2 \land (X,Y) \in (\Q_2^n)^N \times (\Q_2^1)^N}$. By Lemma \ref{MaxCut to DNN}, we have $\set{(\ell^{\DNN}_{1,\tl{X}_G,\tl{Y}_G},1)}{G \in \Graph} \subset S^{\arith}$ and $\MaxCut(\Graph)$ is reduced to $\OP^{\arith}(S)$. Therefore, the NP-hardness of $\MaxCut(\Graph)$ implies that of $S$, and hence that of $\set{\ell^{\DNN}_{1,X,Y}}{(n,m,N) \in \N^3 \land (X,Y) \in (\Q_2^n)^N \times (\Q_2^m)^N}$.
\end{proof}

\begin{thm}[NP-hardness of $p$-adic smooth neural network]
\label{NP-hardness of SNN}
For any prime number $p$, the $p$-adic smooth neural network is NP-hard, i.e. the set
\be
\set{\ell^{\SNN}_{B,X,\vec{y}}}{(D,N) \in \N^2 \land (B,X,\vec{y}) \in B_{p,D} \times (\Qp^D)^N \times \Qp^N}
\ee
of $p$-adic optimisations is NP-hard.
\end{thm}

In order to prove Theorem \ref{NP-hardness of SNN}, we prepare convention and a lemma. Let $G = (n,e) \in \Graph$. We denote by $m \in \N$ the cardinality of the domain of $e$. Set $M \coloneqq m + 2(m+1)n \in \N$ and $E \coloneqq \lceil \log_p \max \ens{1,M} \rceil \in \N$. For each $i \in \N_{< n}$, set
\be
J_{p,G,i} \coloneqq \set{j \in \N_{< M}}{(j < m \to i \in \im(e(j))) \land (j \geq m \to (j - m) \bmod n = i}.
\ee
We define $B_{p,G} \in B_{p,1}$ as the map
\be
\N_{< n} & \to & \CO(\Zp^1) \\
i & \mapsto & \bigsqcup_{j \in J_{p,G,i}} (j + p^E \Zp).
\ee
Set $X_{p,G} \coloneqq (1_{\ens{j}})_{j=0}^{M-1} \in (\Qp^M)^M$. For each $j \in \N_{< M}$, we define $y_{p,G,j} \in \Qp$ as $1$ if $j < m$ and $\lfloor n^{-1}(j-m) \rfloor \bmod 2$ if $j \geq m$. Set $\vec{y}_{p,G} \coloneqq (y_{p,G,j})_{j=0}^{M-1} \in \Qp^M$. Although the convention conflicts that for the proof of Theorem \ref{NP-hardness of LR}, we will not use it anymore.

\begin{lmm}
\label{MaxCut to SNN}
For any $G \in \Graph$, $\ell^{\SNN}_{B_{p,G},X_{p,G},\vec{y}_{p,G}}$ is $(1,\ens{0,1}^n)$-reducible, and a $\vec{c} = (c_i)_{i=0}^{n-1} \in \ens{0,1}^n$ is an optimal solution of $\OP(\ell^{\SNN}_{B_{p,G},X_{p,G},\vec{y}_{p,G}},\ens{0,1}^n)$ if and only if $\set{i \in \N_{< n}}{c_i = 1}$ is an optimal solution of $\MaxCut(G)$.
\end{lmm}

\begin{proof}
We use the same convention of $\vec{\ell}$ and $C$ as in the proof of Lemma \ref{MaxCut to LR}. We have shown the following:
\bi
\item[(1)] The inequality $\sup \Sigma(\vec{\ell})(C) < m + 1$ holds.
\item[(2)] The sequence $\vec{\ell}$ is concentrated on $C$.
\item[(3)] A $(c_i)_{i=0}^{n-1} \in \Qp^n$ is an optimal solution of $\OP(\Sigma(\vec{\ell}),C)$ if and only if $\set{i \in \N_{< n}}{c_i = 1}$ is an optimal solution of $\MaxCut(G)$.
\ei
Let $\vec{c} = (c_i)_{i=0}^{n-1} \in \Qp^n$. For any $j \in \N_{< M}$, we have
\be
\v{y_{p,G,j} - \sum_{i=0}^{n-1} c_i 1_{B_{p,G}(i)}((j))} =
\left\{
\begin{array}{ll}
\ell_j(\vec{c}) & (j < m) \\ 
\Bigv{\left( \left\lfloor \frac{j - m}{n} \right\rfloor \bmod 2 \right) - c_{(j - m) \bmod n}} & (j \geq m)
\end{array}
\right..
\ee
Therefore, we have
\be
\ell^{\SNN}_{B_{p,G},X_{p,G},\vec{y}_{p,G}}(\vec{c}) & = & \Sigma(\vec{\ell})(\vec{c}) + (m+1) \sum_{i=0}^{n-1} \sum_{b=0}^{1} \v{b - c_i},
\ee
i.e.\ $\ell^{\SNN}_{B_{p,G},X_{p,G},\vec{y}_{p,G}}$ is a Baker forcing of $\vec{\ell}$ to $C$, and hence the assertion follows from Theorem \ref{Baker}.
\end{proof}

\begin{proof}[Proof of Theorem \ref{NP-hardness of SNN}]
Set $S \coloneqq \set{\ell^{\SNN}_{B,X,\vec{y}}}{M \in \N \land (B,X,\vec{y}) \in B_{p,1} \times (\Qp^M)^M \times \Qp^M}$. By Lemma \ref{MaxCut to SNN}, we have $\set{(\ell^{\SNN}_{B_{p,G},X_{p,G},\vec{y}_{p,G}},1)}{G \in \Graph} \subset S^{\arith}$ and $\MaxCut(\Graph)$ is reduced to $\OP^{\arith}(S)$. Therefore, the NP-hardness of $\MaxCut(\Graph)$ implies that of $S$, and hence that of $\set{\ell^{\SNN}_{B,X,\vec{y}}}{(D,N) \in \N^2 \land (B,X,\vec{y}) \in B_{p,D} \times (\Qp^D)^N \times \Qp^N}$.
\end{proof}

\vspace{0.3in}
\addcontentsline{toc}{section}{Acknowledgements}
\noindent {\Large \bf Acknowledgements}
\vspace{0.2in}

\noindent
I thank G.\ D.\ Baker for informing me of various techniques in $p$-adic optimisations. I thank all people who helped me to learn mathematics and programming. I also thank my family.

\end{document}